\documentclass[11pt]{article}

\usepackage{amsmath}
\usepackage{amssymb}
\usepackage{amsmath}
\usepackage{amstext}
\usepackage{amsthm}
\usepackage{amsbsy}
\usepackage{amsfonts}
\usepackage{dsfont}
\usepackage{marvosym}
\usepackage{mathtools}
\usepackage{nicefrac}
\usepackage{resizegather} 
\usepackage{authblk}
\usepackage[symbol]{footmisc}
\usepackage[nottoc]{tocbibind}
\usepackage[margin=2.5cm]{geometry}

\newtheorem{thm}{Theorem}[section]
\newtheorem{lem}[thm]{Lemma}

\usepackage{enumitem}

\usepackage[colorlinks=true,linkcolor=black,urlcolor=black,citecolor=black]{hyperref}

\usepackage{graphicx}
\usepackage[section]{placeins}

\usepackage{tabularx}
\usepackage{tabulary}
\usepackage{booktabs}
\usepackage{array}
\usepackage{multirow}
\usepackage{mathdots}
\usepackage{chngcntr}
\counterwithout{figure}{section} 
\counterwithout{table}{section} 

\usepackage{tikz}
\usepackage{tkz-euclide}
\usetikzlibrary{positioning,calc,arrows,shapes,backgrounds}

\usepackage{glossaries}
\glsdisablehyper

\usepackage{pdfpages}

\usepackage[ruled,linesnumbered]{algorithm2e}
\SetKwInOut{Input}{Input}
\SetKwInOut{Output}{Output}
\SetKwInOut{Set}{Set}

\newcommand{\eps}{\varepsilon}

\newcommand{\gamb}{\bar{\gamma}}
\newcommand{\ALGw}{\ensuremath{\mathrm{ALG_{wf}}}\xspace}

\usepackage{csquotes}

\DeclareMathOperator*{\ALG}{ALG}
\DeclareMathOperator*{\OPT}{OPT}
\begin{document}
\begin{center}
{\Large A Stronger Impossibility for Fully Online Matching\footnote[1]{Funding: This work was supported in part by the Alexander von Humboldt Foundation with funds from the German Federal Ministry of Education and Research (BMBF) and by the Deutsche Forschungsgemeinschaft (DFG), GRK 2201.\\
\textit{Email addresses:} \texttt{alexander.eckl@tum.de} (Alexander Eckl), \texttt{anja.kirschbaum@tum.de} (Anja Kirschbaum), \texttt{marilena.leichter@tum.de} (Marilena Leichter), \texttt{kschewior@gmail.com} (Kevin Schewior)}}\\
\par\vspace{\baselineskip}

Alexander Eckl\textsuperscript{1}, Anja Kirschbaum\textsuperscript{1}, Marilena Leichter\textsuperscript{1}, Kevin Schewior\textsuperscript{2}
\par\vspace{\baselineskip}
\footnotesize{\textsuperscript{1} Advanced Optimization in a Networked Economy (AdONE), Technical University of Munich,  Germany\\
\textsuperscript{2} Department of Mathematics and Computer Science,  University of Cologne,  Germany}
\end{center}
\par\vspace{\baselineskip}



\begin{abstract}
We revisit the fully online matching model (Huang et al., J.\ ACM, 2020), an extension of the classic online matching model due to Karp, Vazirani, and Vazirani (STOC 1990), which has recently received a lot of attention (Huang et al., SODA 2019 and FOCS 2020), partly due to applications in ride-sharing platforms. It has been shown that the fully online version is harder than the classic version for which the achievable competitive ratio is at most $0.6317$, rather than precisely $1-\nicefrac1e\approx 0.6321$. We introduce two new ideas to the construction. By optimizing the parameters of the modified construction numerically, we obtain an improved impossibility result of $0.6297$. Like the previous bound, the new bound even holds for fractional (rather than randomized) algorithms on bipartite graphs.
\end{abstract}

\section{Introduction}
\label{sec:introduction}
In the \emph{fully online matching} model due to Huang et al.~\cite{HuangKangEtAl2020a}, the vertices of an undirected graph $G=(V,E)$ arrive over time. Every vertex has a deadline at which it \emph{departs}. At any time, two adjacent vertices can be matched if both of them have arrived, neither of them has departed, and neither of them has been previously matched. An \emph{online} algorithm may base its decisions only on the subgraph of $G$ induced by the vertices that have already arrived, the deadlines of these vertices, and, in the case of a randomized algorithm, random bits.

The classic model due to Karp, Vazirani, and Vazirani~\cite{KarpVaziraniEtAl1990} is a special case. Here, additionally, $G$ is bipartite, initially all vertices from one side arrive, and then the vertices from the other side arrive. The vertices from the former side only depart at the very end. The vertices from the latter side depart even before the next vertex from that side arrives, or, equivalently, they must be matched either upon arrival or never. The classic model and its variants have been extensively studied prior to the work of Huang et al.~\cite{HuangKangEtAl2020a} and have many applications, e.g., in online advertising~\cite{Mehta2013}. However, a scenario not addressed by these models is, e.g., that of ride-sharing platforms in which customers and compatible drivers have to be matched, but both customers and drivers enter and leave the system at arbitrary times. This aspect is addressed by the fully online model.

In both models, the performance of an online algorithm is measured through standard competitive analysis with respect to the total number of matches performed by the algorithm over the entire time horizon: A randomized online algorithm is called $\alpha$-competitive if, for all instances, the expected number of matches it performs is at least an $\alpha$ fraction of the number of matches an omniscient algorithm could have performed. It is well known~\cite{KarpVaziraniEtAl1990,BirnbaumM08} that for the classic model the largest competitive ratio achievable by randomized algorithms is $1-\nicefrac1e\approx 0.6321$, ignoring $o(1)$ terms as $|V|\rightarrow\infty$. It can be obtained through the Ranking algorithm. The same holds true for the \emph{fractional} relaxation of the model in which a (w.l.o.g.\ deterministic) algorithm may fractionally match vertices, thus obtaining a fractional matching~\cite{KalyanasundaramP00}.

\begin{figure}
\centering
\begin{tikzpicture}[xscale=1.3,yscale=.6]
	\draw[->] (1.8,0) -- (13.7,0);	
	\draw (2.13,0.2) -- (2.11,0.2) -- (2.11,-0.2) -- (2.13,-0.2);
	\draw (6.92,0.2) -- (6.90,0.2) -- (6.90,-0.2) -- (6.92,-0.2);
	\draw (9.28,0.2) -- (9.26,0.2) -- (9.26,-0.2) -- (9.28,-0.2);			
	\draw (12.95,0.2) -- (12.97,0.2) -- (12.97,-0.2) -- (12.95,-0.2);
	\draw (13.15,0.2) -- (13.17,0.2) -- (13.17,-0.2) -- (13.15,-0.2);	
	\draw (13.19,0.2) -- (13.21,0.2) -- (13.21,-0.2) -- (13.19,-0.2);		
	\node at (13.7,0.4) {$\alpha$};
	\draw[very thin, draw=black!50] (3,-1) -- (2.11,-0.225);
	\node[anchor=north,align=center] at (3,-1) {\footnotesize $0.5211$\\[-0.7ex] \footnotesize gen.\ int.\ LB\\[-0.7ex] \footnotesize Huang et al.~\cite{HuangKangEtAl2020a}};	
	\draw[very thin, draw=black!50] (6.9,-1) -- (6.9,-0.225);
	\node[anchor=north,align=center] at (6.9,-1) {\footnotesize $0.5690$\\[-0.7ex] \footnotesize bip.\ int.\ LB\\[-0.7ex] \footnotesize Huang et al.~\cite{HuangTangEtAl2020}};
	\draw[very thin, draw=black!50] (9.26,-1) -- (9.26,-0.225);
	\node[anchor=north,align=center] at (9.26,-1) {\footnotesize $0.5926$\\[-0.7ex] \footnotesize gen.\ frac.\ LB\\[-0.7ex] \footnotesize Huang et al.~\cite{HuangTangEtAl2020}};
	\draw[very thin, draw=black!50] (7,1) -- (12.97,0.225);
	\node[anchor=south,align=center] at (7,1) {\footnotesize \textbf{this work} \\[-0.7ex] \footnotesize bip.\ frac.\ UB\\[-0.7ex] \footnotesize $0.6297$};
	\draw[very thin, draw=black!50] (9.5,1) -- (13.17,0.225);
	\node[anchor=south,align=center] at (9.5,1) {\footnotesize Huang et al.~\cite{HuangKangEtAl2020a}\\[-0.7ex] \footnotesize bip.\ frac.\ UB\\[-0.7ex] \footnotesize $0.6317$};
	\draw[very thin, draw=black!50] (12,1) -- (13.21,0.225);
	\node[anchor=south,align=center] at (12,1) {\footnotesize classic~\cite{KarpVaziraniEtAl1990,KalyanasundaramP00} \\[-0.7ex] \footnotesize bip.\ frac.\ UB\\[-0.7ex] \footnotesize $0.6321$};							
\end{tikzpicture}
\caption{Known bounds on the best-possible competitive ratio $\alpha$ for fully online matching. We distinguish between bipartite and general graphs as well as fractional and integral algorithms.}
\label{fig:bounds}
\end{figure}
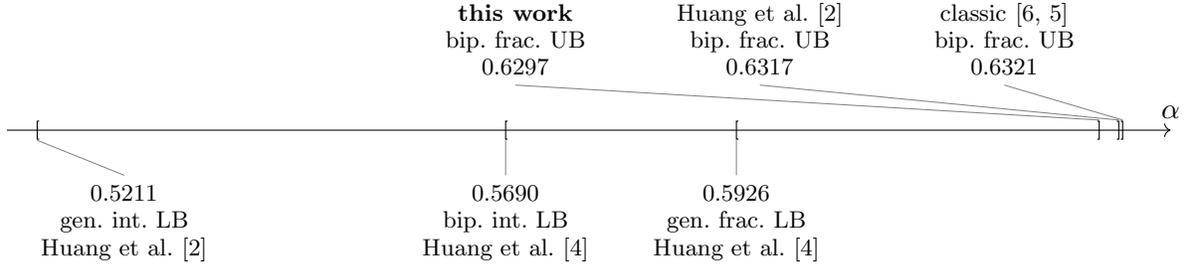

In their seminal work on fully online matching, Huang et al.~\cite{HuangKangEtAl2020a} prove that a generalization of the Ranking algorithm to general graphs is $0.5211$-competitive, beating the trivial baseline of $0.5$, and that no fractional algorithm can achieve a guarantee better than $0.6317$ even on bipartite graphs, also beating the baseline of $1-\nicefrac1e$. In follow-up work, Huang et al.~\cite{HuangPengEtAl2019,HuangTangEtAl2020} revisit the fractional and bipartite cases. The state of the art is a $0.5926$-competitive fractional algorithm on general graphs~\cite{HuangTangEtAl2020} and a $0.5690$-competitive integral (as opposed to fractional) algorithm on bipartite graphs~\cite{HuangTangEtAl2020}. In particular, the impossibility of $0.6317$ is still the state of the art, even for integral algorithms on general graphs. For further related work, we refer to the literature review by Huang et al.~\cite{HuangKangEtAl2020a} and the slightly older (but in-depth) survey by Metha~\cite{Mehta2013}.

\subsection{Our Contribution}

In this work, we give an improvement of the impossibility from $0.6317$ by Huang et al.~\cite{HuangKangEtAl2020a} to $0.6297$. Again, this bound also holds for fractional algorithms and on bipartite graphs. While the construction of Huang et al.\ is not explicitly analyzed for fractional algorithms, the analysis can be easily adapted. In contrast, we explicitly analyze fractional algorithms. We do so without using Yao's principle, avoiding any difficulties that would arise from an infinite strategy space. We show the evolution of this impossibility as well as the state-of-the-art competitive ratios in Figure~\ref{fig:bounds}.

The construction of Huang et al.~\cite{HuangKangEtAl2020a} starts off with presenting a tree level by level starting from the root. The root has degree $\lambda+1$, and exactly one of its children is a leaf.
Which of the children the designated leaf is, however, remains unknown until the root has departed. 
While the optimum can match the leaf vertex to the root, the online algorithm cannot do better than matching the first-level vertices to the root with identical fractional value, leaving some fraction of the leaf vertex permanently unmatched. The construction is then repeated with the $\lambda$ non-leaf children playing the role of the root, for a total of $h$ levels.
Finally, level $h$ of the tree is augmented with the later-arriving side of the \enquote{triangle} construction of Karp, Vazirani, and Vazirani~\cite{KarpVaziraniEtAl1990}. It can be shown that the extra tree puts the algorithm in a worse position and that choosing $\lambda=7$ and $h\rightarrow\infty$ yields the bound of~$0.6317$.

The first ingredient to our result is the observation that this result can be reproduced with a slightly simpler construction: Rather than constructing a tree, one can also construct a sequence of bicliques. Specifically, we start with $k$ vertices (corresponding to the root) on the first level and $(\lambda+1) k$ vertices on the second level, connected by a biclique. After the vertices on the previous level depart, it is revealed which of the $\lambda k$ vertices form a biclique with $(\lambda+1)\lambda k$ new vertices on the next level. Again, this process is repeated for a total of $h$ steps and the \enquote{triangle} construction is added at the end. Note that now we are also allowed to choose any rational value for $\lambda$ as long as we make $k$ large enough. In fact, choosing $\lambda\approx 7.2336$ already slightly
improves the bound by about $9.5\cdot 10^{-7}$.

The observation leading to our second idea is that not all levels play exchangeable roles in the construction. This is clear for the last level, which is part of the triangle construction. Further, we will see the following for an optimal algorithm: Increasing the number of children in any level $i$ changes the average matching value upon departure in any level $j$ with $i<j<h$. Indeed, this matching value decreases if $j-i$ is odd and increases otherwise. This suggests that choosing a uniform $\lambda$ for all levels is unlikely to yield a tight impossibility.

In general, we replace $\lambda$ with different factors $\gamma_1,\dots,\gamma_\ell$ for $\ell$ levels between the first $h$ levels and the \enquote{triangle} construction. After taking the limit $h\rightarrow\infty$, we express the precise resulting competitive ratio as a function of $\lambda, \gamma_1,\dots,\gamma_\ell$. Unfortunately, this function is already non-convex for $\ell=0$, and the explicit formula gets quite complicated even for small values of $\ell$. We therefore have to resort to numerical optimization. Specifically, we obtain the bound of $0.6297$ through numerical optimization with Matlab. In this way, although we cannot prove that $0.6297$ is the strongest impossibility achievable with our approach, the results of the numerical optimization do hint at that, and the claimed impossibility is provably correct.

We note that the numerical values we find for $\lambda,\gamma_1,\dots,\gamma_\ell$ back our intuition that the optimal value heavily depends on whether the corresponding level is the last one, and otherwise its parity.


\section{Description of the Construction}
\label{sec:description}
In this section, we formally describe the construction that is used to show the improved impossibility. The construction is parameterized by integers $h,\ell\geq 0$ and $k\geq 1$ as well as fractional numbers $\lambda>1$ and $\gamma_j>1$ for all $j\in[\ell]\coloneqq\{1,\dots,\ell\}$. We first define the bipartite graph $G$, visualized in Figure~\ref{fig:ex}, underlying the construction and then specify the arrival and departure times of the vertices.

The vertex set of $G$ can be partitioned into $2\cdot(h{+}\ell{+}1)$ disjoint sets $U_i,V_i$ for all $i\in\{0,\dots,h{+}\ell\}$. We have $|U_i|=|V_i|$ for each such $i$. The set $U_0$ consists of $k$ vertices. With increasing index, the number of vertices increases for the first $h$ steps by a factor of $\lambda$ and in the $j$-th subsequent step by a factor of $\gamma_j$. Formally,
\begin{equation*}
	\begin{aligned}
		|U_i| 		&= \lambda^i k 			&& \forall i \in  \{0,\dots,h\}, \\
		|U_{h+j}| 	&= \gamma_j |U_{h+j-1}|	&& \forall j \in \{1,\dots,\ell\}. \\	
	\end{aligned}
\end{equation*}
We note that by the continuous nature of the variables $\gamma_j$ and $\lambda$, some of the sizes $|U_i| = |V_i|$ might take non-integer values. We avoid this by choosing the parameter $k$ in dependence of $h$ large enough such that all of these numbers are integer.

We introduce edges such that $(U_i,V_i\cup U_{i+1})$ is a complete bipartite graph for all $i \in \{0,\dots,h{+}\ell{-}1\}$. The edges between $U_{h{+}\ell}$ and $V_{h{+}\ell}$ form an exception. Due to this exception, we define $A\coloneqq U_{h{+}\ell}$ and $B \coloneqq V_{h{+}\ell}$. In Section~\ref{sec:proof}, we describe how the adversary chooses an ordering of $A=\{a_1,\dots,a_{|A|}\}$ and $B=\{b_1,\dots,b_{|B|}\}$. Then $a_i$ and $b_j$ are connected by an edge if and only if $i\geq j$. The subgraph of $G$ induced by $A\cup B$ has been previously used to show impossibilities for online matching problems~\cite{KarpVaziraniEtAl1990,KalyanasundaramP00}. This completes the description of $G$.

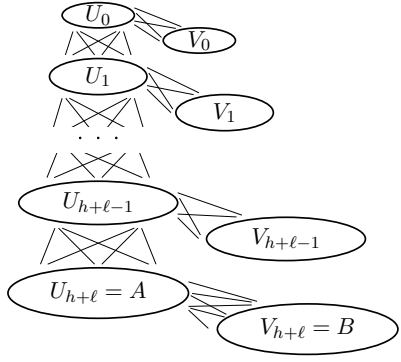
\begin{figure}
\begin{center}
\def\Vset[#1] (#2) (#3,#4)
  { \draw[#1, thick] (#2) circle  [x radius=#3cm, y radius=#4cm]; }
\scalebox{.8}{
\begin{tikzpicture}[scale=.6, every vertex/.style={scale=0.8}]

\edef\x{3}
\edef\xu{5}

\edef\y{8.7}
\edef\rx{1}
\edef\ry{.35}
\Vset[black] (\x,\y) (\rx,\ry) 
\draw (\x,\y) node {$U_0$};

\edef\yu{8}
\edef\xu{\x+2*\rx+.8}
\Vset[black] (\xu,\yu) (\rx,\ry) 
\draw (\xu,\yu) node {$V_0$};

\edef\crox{\xu-.6*\rx}
\edef\croy{\yu+1.1*\ry}
\edef\clox{\x+1*\rx}
\edef\cloy{\y+.3*\ry}
\edef\clux{\x+\rx}
\edef\cluy{\y-.6*\ry}
\edef\crux{\xu-1.1*\rx}
\edef\cruy{\yu+.1*\ry}
\draw[] (\clox,\cloy+.1)--(\crox+.1,\croy);
\draw[] (\clux,\cluy+.05)--(\crox-.05,\croy-.05);
\draw[] (\clux,\cluy-.05)--(\crux,\cruy-.05);
\draw[] (\clox+.05,\cloy-.05)--(\crux+.05,\cruy+.05);

\edef\yb{7}
\edef\rxb{1.35}
\edef\ryb{.5}
\Vset[black] (\x,\yb) (\rxb,\ryb) 
\draw (\x,\yb) node {$U_1$};

\edef\yu{6}
\edef\xu{\x+2*\rxb+.8}
\Vset[black] (\xu,\yu) (\rxb,\ryb) 
\draw (\xu,\yu) node {$V_1$};

\edef\crox{\xu-.6*\rxb}
\edef\croy{\yu+1.1*\ryb}
\edef\clox{\x+1*\rxb}
\edef\cloy{\yb+.3*\ryb}
\edef\clux{\x+\rxb}
\edef\cluy{\yb-.6*\ryb}
\edef\crux{\xu-1.1*\rxb}
\edef\cruy{\yu+.1*\ryb}
\draw[] (\clox,\cloy+.1)--(\crox+.1,\croy);
\draw[] (\clux,\cluy+.05)--(\crox-.05,\croy-.05);
\draw[] (\clux,\cluy-.05)--(\crux,\cruy-.05);
\draw[] (\clox+.05,\cloy-.05)--(\crux+.05,\cruy+.05);

\edef\oxa{\x-.6*\rx}
\edef\oxb{\x+.6*\rx}
\edef\oy{\y-\ry-.1}
\edef\uxa{\x-.6*\rxb}
\edef\uxc{\x}
\edef\uxb{\x+.6*\rxb}
\edef\uy{\yb+\ryb+.1}
\draw[] (\oxa-.1,\oy)--(\uxa-.1,\uy);
\draw[] (\oxb+.1,\oy)--(\uxb+.1,\uy);
\draw[] (\oxa+.1,\oy)--(\uxb-.1,\uy);
\draw[] (\oxb-.1,\oy)--(\uxa+.1,\uy);
\draw[] (\oxa,\oy)--(\uxc-.1,\uy);
\draw[] (\oxb,\oy)--(\uxc+.1,\uy);

\edef\y{7}
\edef\rx{1.5}
\edef\ry{.5}
\edef\yb{4.5}
\edef\rxb{2}
\edef\ryb{.8}

\edef\oxa{\x-.6*\rx}
\edef\oxb{\x+.6*\rx}
\edef\oy{\y-\ry-.1}
\edef\uxa{\x-.6*\rxb}
\edef\uxc{\x}
\edef\uxb{\x+.6*\rxb}
\edef\uy{\yb+\ryb+.1}
\draw[] (\oxa-.1,\oy)--(\uxa-.1,\uy);
\draw[] (\oxb+.1,\oy)--(\uxb+.1,\uy);
\draw[] (\oxa+.1,\oy)--(\uxb-.1,\uy);
\draw[] (\oxb-.1,\oy)--(\uxa+.1,\uy);
\draw[] (\oxa,\oy)--(\uxc-.1,\uy);
\draw[] (\oxb,\oy)--(\uxc+.1,\uy);

\edef\y{5.8}
\edef\rx{1.6}
\edef\ry{.6}
\edef\yb{3.5}
\edef\rxb{2.2}
\edef\ryb{.6}
\Vset[black] (\x,\yb) (\rxb,\ryb) 
\draw (\x,\yb) node {$U_{h+\ell-1}$};
\edef\yu{2.5}
\edef\xu{\x+2*\rxb+.8}
\Vset[black] (\xu,\yu) (\rxb,\ryb) 
\draw (\xu,\yu) node {$V_{h+\ell-1}$};

\edef\crox{\xu-.6*\rxb}
\edef\croy{\yu+1.1*\ryb}
\edef\clox{\x+1*\rxb}
\edef\cloy{\yb+.3*\ryb}
\edef\clux{\x+\rxb}
\edef\cluy{\yb-.6*\ryb}
\edef\crux{\xu-1.1*\rxb}
\edef\cruy{\yu+.1*\ryb}
\draw[] (\clox,\cloy+.1)--(\crox+.1,\croy);
\draw[] (\clux,\cluy+.05)--(\crox-.05,\croy-.05);
\draw[] (\clux,\cluy-.05)--(\crux,\cruy-.05);
\draw[] (\clox+.05,\cloy-.05)--(\crux+.05,\cruy+.05);

\edef\oxa{\x-.6*\rx}
\edef\oxb{\x+.6*\rx}
\edef\oy{\y-\ry-.1}
\edef\uxa{\x-.6*\rxb}
\edef\uxc{\x}
\edef\uxb{\x+.6*\rxb}
\edef\uy{\yb+\ryb+.1}
\draw[] (\oxa-.1,\oy)--(\uxa-.1,\uy);
\draw[] (\oxb+.1,\oy)--(\uxb+.1,\uy);
\draw[] (\oxa+.1,\oy)--(\uxb-.1,\uy);
\draw[] (\oxb-.1,\oy)--(\uxa+.1,\uy);
\draw[] (\oxa,\oy)--(\uxc-.1,\uy);
\draw[] (\oxb,\oy)--(\uxc+.1,\uy);

\fill[fill=white](1.5,\y-.2)--(4.5,\y-.2)--(4.5,\y-.9)--(1.5,\y-.9)--cycle;
\filldraw (\x-.5,\y-.5) circle (.6pt);
\filldraw (\x,\y-.5) circle (.6pt);
\filldraw (\x+.5,\y-.5) circle (.6pt);

\edef\y{3.5}
\edef\rx{2.2}
\edef\ry{.6}
\edef\yb{1}
\edef\rxb{2.5}
\edef\ryb{.7}

\Vset[black] (\x,\yb) (\rxb,\ryb) 
\draw (\x,\yb) node {$U_{h+\ell}=A$};

\edef\yu{0}
\edef\xu{\x+2*\rxb+.8}
\Vset[black] (\xu,\yu) (\rxb,\ryb) 
\draw (\xu,\yu) node {$V_{h+\ell}=B$};

\edef\oxa{\x-.5*\rx}
\edef\oxb{\x+.5*\rx}
\edef\oy{\y-\ry-.1}
\edef\uxa{\x-.6*\rxb}
\edef\uxc{\x}
\edef\uxb{\x+.6*\rxb}
\edef\uy{\yb+\ryb+.1}
\draw[] (\oxa-.1,\oy)--(\uxa-.1,\uy);
\draw[] (\oxb+.1,\oy)--(\uxb+.1,\uy);
\draw[] (\oxa+.1,\oy)--(\uxb-.1,\uy);
\draw[] (\oxb-.1,\oy)--(\uxa+.1,\uy);
\draw[] (\oxa,\oy)--(\uxc-.1,\uy);
\draw[] (\oxb,\oy)--(\uxc+.1,\uy);

\edef\crox{\xu-.6*\rxb}
\edef\croy{\yu+1.1*\ryb}
\edef\clox{\x+1*\rxb}
\edef\cloy{\yb+.3*\ryb}
\edef\clux{\x+\rxb}
\edef\cluy{\yb-.6*\ryb}
\edef\crux{\xu-1.1*\rxb}
\edef\cruy{\yu+.1*\ryb}
\edef\crmx{\xu-.9*\rxb}
\edef\crmy{\yu+.55*\ryb}
\draw[] (\clox,\cloy+.1)--(\crox+.1,\croy);
\draw[] (\clux,\cluy+.05)--(\crox-.05,\croy-.1);
\draw[] (\clux,\cluy-.05)--(\crux,\cruy-.05);
\draw[] (\clux+.05,\cluy)--(\crmx-.05,\crmy-.05);

\draw[] (\clox+.06,\cloy-.28)--(\crmx,\crmy+.025);
\draw[] (\clox+.1,\cloy-.22)--(\crox-.05,\croy-.02);

\end{tikzpicture}}
\end{center}
\caption{Visualization of the bipartite graph $G$.}\label{fig:ex}
\end{figure}

To complete the description of our \emph{online} construction, we specify the arrival and departure times of the vertices in the graph $G$. At the beginning, all vertices in $U_0$ arrive simultaneously. Let $i \in \{0,\dots,h{+}\ell{-}1\}$ and assume all vertices in $U_j$ for $j\leq i$ and in $V_j$ for $j\leq i-1$ have already arrived, and all of these vertices not in $U_i$ have departed again.
Then, all vertices in $U_{i+1}$ and $V_i$ (by definition, along with their edges to $U_i$) arrive simultaneously. Next, the vertices in $U_i$ and $V_{i-1}$ (if $i\geq 1$) depart simultaneously. At this point, it is impossible for an algorithm to differentiate between the neighbors of $U_i$, in particular identifying which of them belong to $U_{i+1}$ and which to $V_i$. In Section~\ref{sec:proof}, we describe how the adversary can assign these vertices to $U_{i+1}$ and $V_i$ based on the matching decisions of the algorithm.
If $i\leq h{+}\ell{-}2$, the initial assumption is reestablished for the next-larger index, and the arrivals and departures happen as described. Otherwise, $A=U_{h{+}\ell}$ has now arrived, and all other previously arrived vertices have departed. Then the vertices in $B$ arrive in order of $b_1,\dots,b_{|B|}$, and every vertex departs again immediately after its arrival. Finally, the vertices $A$ and $V_{h+\ell-1}$ depart simultaneously.
Note that all edges of the graph respect release and departure times in the sense that, for all of them, there is a time when both their endpoints have arrived but not departed.

The graph has a perfect matching which is the union of the following perfect matchings: From the complete bipartite subgraphs induced by  $U_i \cup V_i$ we choose an arbitrary perfect matching. For $A$ and $B$ we choose the perfect matching containing $(a_i,b_i)$ for all $i \in [|A|]$. We also remark that $U \coloneqq \bigcup_{i=0}^{h+\ell} U_i$ and $V \coloneqq \bigcup_{i=0}^{h+\ell} V_i$, however, do \emph{not} correspond to the two sides of the bipartite graph.

\section{Derivation of the Upper Bound}
\label{sec:proof}

Using the construction from the previous section, we show our result in this section.

\begin{thm}\label{thm:main}
	The competitive ratio of any fractional algorithm for fully online matching is at most $0.6297$, even on bipartite graphs.
\end{thm}

To prove this result, we have to bound both the value of the offline optimum and that of any online algorithm. Clearly, since $G$ has a perfect matching
and its vertex set is $U\cup V$ where $|U|=|V|$, we have $\OPT=|U|=|V|$.

Now let $\ALG$ be any fractional online algorithm. For all vertices in our graph, we determine the fraction to which they are matched. For simplicity, we make two assumptions that are without loss of generality. First, we assume that the algorithm only performs matches along an edge whenever one of its endpoints departs. Furthermore, when a vertex departs, we assume that the algorithm fully matches it unless all of its neighbors are fully matched. It can be seen by a simple exchange argument that the latter assumption is indeed without loss of generality (see also~\cite{HuangKangEtAl2020a}).

Note that we only need to consider the matching value distributed at the departure times of vertices in $U_0\cup\dots\cup U_{h+\ell-1}\cup B$. We do so in order of departure of the respective vertices, so we start with $U_{i-1}$ for $i \in \{1,\dots,h{+}\ell\}$.
Recall that, when these vertices reach their deadlines, their \emph{not departed} neighbors $N^+(U_{i-1}) \coloneqq U_i \cup V_{i-1}$ are indistinguishable for the algorithm. For simplicity, we would like to analyze the water-filling algorithm~\cite{KalyanasundaramP00} $\ALGw$ which also treats these vertices equally, i.e., matches them all to the same fraction. It may, however, conceivably help the algorithm to match vertices to \emph{different fractions}. In the following, we first show that this is (essentially) not the case.

We will define an adversary that assigns the vertices in $N^+(U_{i-1})$ to $U_i$ and $V_{i-1}$ (respecting $|U_{i-1}|=|V_{i-1}|$) based on the matching decisions of the algorithm. A first approach may be to just assign those vertices to $V_{i-1}$ that have been matched the least. We will, however, later need a lower \emph{and an upper} bound on the value to which $U_i$ has been matched. The latter approach only gives us the lower bound. 

We define $p_{i-1}$ such that the total matching value of the vertices in $U_{i-1}$ just after the departure of $U_{i-2}$ is ${p_{i-1}\cdot |U_{i-1}|}$. In other words, $p_{i-1}$ is the average matching value across all (again, conceivably different) matching values of vertices in $U_{i-1}$ at that point. We define $p_0 = 0$.

Note that, upon departure of $U_{i-1}$, \ALGw would assign a total fractional matching value of
\begin{equation}\label{eq:avg}
	d \coloneqq (1-p_{i-1}) \cdot |U_{i-1}|\ \frac{|U_i|}{|N^+(U_{i-1})|}
\end{equation}
to any set of $|U_i|$ vertices in $N^+(U_{i-1})$. Our adversary simply chooses an assignment of $N^+(U_{i-1})$ to $U_i$ and $V_{i-1}$ such that $p_i\cdot|U_i|$
is as close as possible to $d$. The following lemma shows that the error becomes vanishingly small.

\begin{lem}
\label{lem:match_val_distr_U}
Let $i \in [h{+}\ell]$ and define $n_i\coloneqq |U_i|$. For any given distribution of matching value to $N^+(U_{i-1})$ by the algorithm, there is an adversary that chooses sets $U_i$ and $V_{i-1}$ such that
		\begin{align*}
			p_i &= \frac{1-p_{i-1}}{\lambda+1} + \eps_1(i),&&\text{for }i\leq h \;\;\; \text{and}\\  \qquad p_i&=\frac{1-p_{i-1}}{\gamma_j+1} + \eps_1(h{+}j), &&\text{for } i=h{+}j,j\in[\ell],
		\end{align*}
		where the error term $\eps_1(i)$ fulfills $|\eps_1(i)| \le \frac{1}{n_i}.$
\end{lem}

\begin{proof}
	For now, let $i \in [h]$ and let a distribution of the matching value from $U_{i-1}$ to $N^+(U_{i-1})$ be given. Let $m(U_i)$ and $m(V_{i-1})$ be the total matching value received by $U_i$ and $V_{i-1}$, respectively. 
	The adversary decides on a partition $U_i\;\dot{\cup}\;V_{i-1}$ of $N^+(U_{i-1})$. 
	We let the adversary partition the vertices such that $m(U_i)$ is as close as possible to $d$ as in Equation~\eqref{eq:avg}. Note that here
	\begin{equation*}
	d= (1-p_{i-1}) |U_{i-1}|\ \frac{|U_i|}{|N(U_{i-1})|} = |U_i|\ \frac{1-p_{i-1}}{\lambda+1}.
	\end{equation*}
	Note that this value does not depend on the explicit partition since $|U_i|$ has a prescribed size which is not chosen by the adversary. It is clear that $d$ corresponds to the total matching value assigned to $U_i$ if all vertices in $N^+(U_{i-1})$ have received the same amount of matching value. It turns out that, independently of the assignment by the algorithm, it is always possible to assign sets $U_i$ and $V_{i-1}$ with the correct cardinalities such that 
	\begin{equation}
	\label{eq:adv_chosen_set_U_i}
	d - 1 \le m(U_i) \le d + 1.
	\end{equation}
	To verify this claim, assume for contradiction that $\hat{U}_i$ is the set with minimal distance $\left|m(\hat{U}_i) - d \right| > 1.$ Additionally, let $\hat{V}_{i-1} = N^+(U_{i-1}) \setminus \hat{U}_i$ be the complementary set in the partition.
	
	Firstly, assume that $m(\hat{U}_i)$ is larger than $d$.	Let $u \in \hat{U}_i$ be the element of $\hat{U}_i$ with maximum matching value and let $v \in \hat{V}_{i-1}$ be the element in $\hat{V}_{i-1}$ with minimum matching value.
	It must hold that $1 \ge m(u) > m(v) \ge 0$, otherwise the matching $m(\hat{U}_i)$ cannot exceed the weighted average $d$. 
	This implies $0<m(u)-m(v)\leq 1$.
	Hence when we exchange $u$ and $v$ in the sets $\hat{U}_i$ and $\hat{V}_{i-1}$, we have $m(\hat{U}_i \setminus \{u\} \cup \{v\}) - d = m(\hat{U}_i)-m(u)+m(v) - d$ and $0\leq m(\hat{U}_i)-d-(m(u)-m(v))< m(\hat{U}_i)-d,$
	so $\left|m(\hat{U}_i \setminus \{u\} \cup \{v\}) - d \right| < \left|m(\hat{U}_i) - d \right|,$
	a contradiction to the minimality of the distance.
	
	Secondly, when $m(\hat{U}_i)$ is smaller than $d$, let $u \in \hat{U}_i$ be the element of $\hat{U}_i$ with minimum matching value and let $v \in \hat{V}_{i-1}$ be the element in $\hat{V}_{i-1}$ with maximum matching value. It holds $0 \le m(u) < m(v) \le 1$, and hence we can again switch $u$ and $v$ to receive $\left|m(\hat{U}_i \setminus \{u\} \cup \{v\}) - d \right| < \left|m(\hat{U}_i) - d \right|$. Therefore, we have proven Equation \eqref{eq:adv_chosen_set_U_i} by contradiction. We divide it by $|U_i|$ to receive
	\begin{equation*}
	\frac{1-p_{i-1}}{\lambda+1} - \frac{1}{|U_i|} \le \frac{m(U_i)}{|U_i|} \le \frac{1-p_{i-1}}{\lambda+1} + \frac{1}{|U_i|}.
	\end{equation*}
	By the definition of $p_i$ as $m(U_i)/|U_i|$, we finally have	
	$$p_i = \frac{1-p_{i-1}}{\lambda+1} +\eps_1(i)$$ for some error term with $|\eps_1(i)|\leq \frac{1}{n_i}$.
	
	We repeat the above arguments for $j \in [\ell]$ with the corresponding growth parameters to obtain 
	$$p_{h+j} = \frac{1-p_{h+j-1}}{\gamma_j+1} +\eps_1(h{+}j)$$ with $|\eps_1(h{+}j)|\leq \frac{1}{n_{h+j}}.$
\end{proof}	

Via induction, we obtain a closed-form description of $p_i$ (see appendix):
\begin{equation}
\label{eq:closed_pi}
	p_i = \frac{1}{\lambda+2}\left(1-\left(\frac{-1}{\lambda+1}\right)^i\right) + \eps_2(i), \qquad \forall i \in [h],
\end{equation}
with $|\eps_2(i)| \le \frac{i}{n_i}$. Computing explicit formulas for $p_{j}, j >h$ is quite cumbersome. We detail the algebraic transformations for $\ell=3$ in the appendix.

We are interested in the matching value the algorithm gives to vertices in $V$. Let us define $q_i$ as the total matched fraction of the vertices in $V_i$ for all $i \in \{0,\dots,h{+}\ell{-}1\}$. Note that $B = V_{h+\ell}$ is matched differently, something we will consider at a later point.

\begin{lem}
\label{lem:match_val_distr_V}
Let $i \in \{0,\dots,h{+}\ell{-}1\}$. Then it holds that
	\begin{equation*}
		q_i = \overline{p}_{i+1} + \eps_3(i),
	\end{equation*}
	where $\overline{p}_{i+1} \coloneqq p_{i+1} - \eps_2(i{+}1)$ and $|\eps_3(i)|\leq \frac{i+3}{n_i}$.
\end{lem}

\begin{proof}
	Again, we only prove the cases $i \in \{0,\dots,h{-}1\}$; the cases of larger $i$ are analogous with adapted growth parameters. Let a distribution of the matching value from $U_i$ to their neighbors be given. Since the total matching value in $U_i$ is $(1-p_i)n_i$, it holds that
	\begin{equation*}
	(1-p_i)n_i = p_{i+1}n_{i+1} + q_i n_i.
	\end{equation*}
	Rearranging this equation, we use the formula from Lemma \ref{lem:match_val_distr_U} to write
	\begin{align*}
	q_i n_i 	&= (1-p_i)n_i - p_{i+1}n_{i+1} \\
	&= (1-p_i)n_{i}  - \left( \frac{1-p_i}{\lambda+1}+\eps_1(i{+}1)\right) \lambda n_i\\
	&= n_i \left( \frac{1-p_i}{\lambda+1} -\lambda \eps_1(i{+}1) \right).
	\end{align*}
	With $n_i=|U_i| = |V_i|$ we can simplify to $q_i = \frac{1-p_i}{\lambda+1} -\lambda \eps_1(i{+}1).$ Inserting Lemma \ref{lem:match_val_distr_U} gives
	\begin{align*}
		q_i &= p_{i+1} - (\lambda+1) \eps_1(i{+}1)\\
			&= (p_{i+1} - \eps_2(i{+}1)) + (\eps_2(i{+}1) - (\lambda+1) \eps_1(i{+}1))\\
			&= \overline{p}_{i{+}1} + \eps_3(i),
	\end{align*}
	where it is easy to see that $|\eps_3(i)| \leq \frac{i+3}{n_i}$.
\end{proof}

In other words, we express the $q_i$ in terms of $p_{i+1}$ without the error terms $\eps_2(i{+}1)$ for all $i \in \{0,\dots, h{+}\ell{-}1\}$.

We can write the amount of matching value the vertex sets $V_i$ (except $B$) contribute to the algorithmic value as
\begin{equation*}
	\sum_{i=0}^{h+\ell-1} q_i|V_i|.
\end{equation*}

Finally, we consider the vertices in $A$. Recall that their total matching value is $p_{h+\ell}\cdot |A|$ before any vertex in $B$ has arrived. For simplicity, we define $p_A \coloneqq p_{h+\ell}$. We also define $\rho$ such that the additional matching value that $A$ receives (that is, at the departure times of vertices in $B$) is $\rho\cdot |A|$. In the following, we will give an upper bound on $\rho$.

Again, an analysis of \ALGw would be simpler but some carefulness is required because the matching value $p_A\cdot |A|$ is not necessarily distributed uniformly across $A$. Fortunately, a simple adversary for choosing the ordering of $A$ and $B$ suffices here.

\begin{lem}
\label{lem:match_val_final_level}
	The adversary can choose the ordering of $A$ and $B$ such that
	\begin{equation*}
		\rho \le 1- \exp({-(1- p_A)}) + \frac{2}{|A|}.
	\end{equation*}
\end{lem}

\begin{proof}
	As described in Section \ref{sec:description}, the vertices in $B$ arrive, and immediately depart again, sequentially.
	The vertices are labeled $b_1,\dots,b_{|B|}$ in this order, and any vertex $b_i$ is adjacent to the vertices $a_i, \dots, a_{|A|}$.
	Note that the algorithm only learns about the identity of vertex $a_i$ after the departure of $b_{i}$.
	Here, the adversary simply chooses $a_i$ in every round such that it has the minimum current matching value out of all remaining unlabeled vertices in $A$.
	
	We now bound $\rho \cdot|A|$, the fractional matching value placed by the algorithm on edges between $A$ and $B$, by the matching value $\omega \cdot|A|$ placed by \ALGw on the same instance: Assuming that the matching value in $A$ is equally distributed before $B$ arrives, i.e., all vertices have \emph{exactly} $p_A$ matching value, this algorithm simply matches $b_i$ equally among all vertices $a_i, \dots, a_{|A|}$.
	We claim that $\rho \leq \omega$. In the following, we denote the matching value placed on $a\in A$ by our algorithm and \ALGw at the time of the departure of $b_i$ by $m_i(a)$ and $\omega_i(a)$, respectively. We omit the index $i$ if we refer to the final value.
	
	Let $\eta\in [|B|]$ the final index for which \ALGw is able to assign matching value to $A$. Note that every time $i$  appears before $b_\eta$, $\omega_i(A)\geq m_i(A)$, as \ALGw always matches the current vertex $b_i$ fully if possible. Also, for $i \geq \eta$ it holds that $1 = \omega(a_i)\geq m(a_i)$, so $\omega(\{a_i:i\geq \eta\})\geq m(\{a_i:i\geq \eta\})$. Now let $i'<\eta$ be the maximal index $i$ so that $\omega(a_i)< m(a_i)$. By the adversary's choice of $a_{i'}$, it holds that
	\begin{align*}
	m_{i'}(\{a_i:i>i'\})&\geq m(a_{i'})|\{a_i:i>i'\}|\\
	&>\omega(a_{i'})|\{a_i:i>i'\}|=\omega_{i'}(\{a_i:i>i'\}).
	\end{align*}
	Using
	\begin{align*}
		&\omega(\{a_i:i\leq i'\})+\omega_{i'}(\{a_i:i>i'\})=\omega_{i'}(A)\\
		&\ge m_{i'}(A)=m(\{a_i:i\leq i'\})+m_{i'}(\{a_i:i>i'\}),
	\end{align*}
	we have $\omega(\{a_i:i\leq i'\})\geq m(\{a_i:i\leq i'\})$ and as $i'<\eta$ is the last vertex with $\omega(a_i) < m(a_i)$, we know that for all remaining vertices \ALGw assigns more matching value. Therefore, $\omega(\{a_i:i'<i< \eta\})\geq m(\{a_i:i'<i<\eta\})$ and thus $\omega(A)\geq m(A)$. Finally, we note that if $i'$ does not exist, $\omega(\{a_i:i< \eta\})\geq m(\{a_i:i<\eta\})$ holds trivially.
	
	All that remains to prove the lemma is to show that $\omega \le 1-\exp({-(1-p_A)})$. Let again $\eta \in [|B|]$ be the final index for which \ALGw is able to assign value to $A$. At the departure of all vertices $b_i,\, i < \eta$ in $B$, \ALGw assigns $\nicefrac{1}{(|A|-i+1)}$ of matching value each to $a_i,\dots,a_{|A|}$, while at the departure of $b_i,\, i > \eta$ no value is assigned. 
	So the average fractional matching value $\omega$ fulfills $\omega |A| \le \eta$, which holds with equality exactly if the last active vertex $b_\eta$ is fully matched.
		
	Let us again refer to the matching value placed by \ALGw on a vertex $a_i$ by $\omega(a_i)$. After $b_\eta$ has departed, all remaining vertices $a_\eta,\dots,a_{|A|}$ are already fully matched. Hence, for all $i \ge \eta$, we have $\omega(a_i) = 1$. At the same time, $\omega(a_i), i \ge \eta$ is composed of the initial value $p_A$ plus the entire value assigned by the vertices $b_1,\dots,b_\eta$. For all $i \ge \eta$ we can obtain:
	\begin{align*}
	1 = \omega(a_i)&\ge p_A + \sum_{j=1}^{\eta-1} \frac{1}{(|A|-j+1)}\\
	&= p_A + H_{|A|} - H_{|A|-\eta+1},
	\end{align*}
	where $H_n$ is the $n$-th harmonic number which we estimate by $\ln(n) \le H_n \le \ln(n+1)$. Hence, we have
	\begin{align*}
	1 &\ge p_A + \ln(|A|) - \ln(|A|-\eta+2) \ge p_A + \ln \left( \frac{|A|}{|A|(1-\omega) +2} \right),
	\end{align*}
	where we used $\omega |A| \le \eta$. By further rearranging,
	\begin{alignat*}{2}
		\exp({-(1- p_A)}) 	&\le  1 - \exp({-(1- p_A)})  + \frac{2}{|A|}.
	\end{alignat*}
	This completes the proof.
\end{proof}

We have now computed all necessary values for our formula. We double-count the matching value by counting the fractional value to which each \emph{vertex} is matched and establish for the entire value of ALG:
\begin{equation*}
2 \ALG = \sum_{i=0}^{h+2} q_i|V_i| + |U\setminus A| + p_A|A| + 2 \rho |A|.
\end{equation*}

For $\ell=0$, by inserting our formulas into $\nicefrac{\ALG}{\OPT}$ and taking the limit $h\to\infty$, we obtain in congruence with Huang et al.~\cite{HuangKangEtAl2020a}
\begin{align*}
		\frac{\lambda-1}{\lambda}\cdot\left(1-\exp\left(-\frac{\lambda+1}{\lambda+2}\right)\right)+\frac{\lambda+1}{\lambda\cdot(\lambda+2)}
\end{align*}
as an upper bound on the competitive ratio. Interestingly, this function is non-convex as its derivative has a local maximum at $\lambda\approx 10.0266$.

\begin{table*}
\centering
\label{tab:results}
\begin{tabular}{ccccccccc} 
\toprule
$\ell$ & $\lambda$ & $\cdots$ & $\gamma_{\ell-4}$ & $\gamma_{\ell-3}$ & $\gamma_{\ell-2}$ & $\gamma_{\ell-1}$ & $\gamma_{\ell}$ & impossibility\ \\
\midrule 
0 & 7.233629 & $\cdots$ & -- & -- & -- & -- & -- & 0.631744 \\
1 & 2.581174 & $\cdots$  & -- & -- & -- & -- &  8.053197 & 0.629748 \\
2 & 3.148324 & $\cdots$ & -- & -- & -- & 2.390115 & 7.874599 & 0.629678 \\
3 & 2.875859 & $\cdots$ & -- & -- & 3.249854 & 2.403421 & 7.864072 & 0.629674 \\
$\vdots$ &$\vdots$ & $\ddots$ & $\vdots$ & $\vdots$& $\vdots$ & $\vdots$ &$\vdots$ & $\vdots$ \\
10 & 2.94419 & $\cdots$ & 2.986001 & 2.843101 & 3.241640 &  2.404098 & 7.863523 & 0.629674 \\
\bottomrule
\end{tabular}
\caption{The results of the numerical optimization. All numbers are rounded to the sixth decimal digit.}
\end{table*}

Computing explicit formulas for all $q_j$ when $j>h$ is quite cumbersome, and the same holds for the explicit formula of the resulting lower bound on the competitive ratio. We showcase the result of this computation for $\ell=3$. We refer to the appendix for all details on the computation. The error terms in the formula vanish as we take the limit $h \to \infty$.
The final formula for our upper bound on the competitive ratio depends only on $\gamma_1, \gamma_2, \gamma_3$ and $\lambda$:

\begin{align*}
&\frac{\lambda+\gamma_1(\gamma_2+1)(\lambda-1)}{2(\lambda+\gamb(\lambda-1))}\\
+&\frac{\lambda^2+\gamma_1}{2(\lambda+2)(\gamma_1+1)(\lambda+\gamb(\lambda-1))}\\
+&\frac{\gamma_1(\lambda-1)}{2(\lambda+\gamb(\lambda-1))}\cdot\frac{\gamma_1(\lambda+2)+1}{(\gamma_2+1)(\gamma_1+1)(\lambda+2)}\\
+&\frac{\gamma_1\gamma_2(\lambda-1)}{2(\lambda+\gamb(\lambda-1))}\cdot\frac{\gamma_2(\gamma_1+1)(\lambda+2)+(\lambda+1)}{(\gamma_2+1)(\gamma_1+1)(\lambda+2)}\\
+&\frac{\gamma_1\gamma_2\gamma_3(\lambda-1)}{(\lambda+\gamb(\lambda-1))} \, \cdot \\
&\quad \Bigg[1-\exp\left(-\frac{1}{\gamma_3+1}\left(\gamma_3{+}\frac{\gamma_1(\lambda+2)+1}{(\gamma_2+1)(\gamma_1+1)(\lambda+2)}\right)\right)\Bigg],
\end{align*}
using the abbreviation $\gamb\coloneqq \sum_{i=1}^3 \prod_{j=1}^i \gamma_j$.\\
We use the numerical computing software Matlab for the numerical optimization. Using a trusted-region algorithm for unconstrained multivariate minimization we receive the values shown in Table~\ref{tab:results}.
In particular, we obtain an upper bound smaller than $0.6297$. This completes the proof of Theorem~\ref{thm:main}.

\bibliographystyle{abbrv}
\bibliography{UpperBoundOnlineMatching}

\appendix
\section{Computing closed forms for $p_i$}

We use induction to obtain the closed form
\begin{equation*}
	p_i = \frac{1}{\lambda+2}\left(1-\left(\frac{-1}{\lambda+1}\right)^i\right) +\eps_2(i), \qquad \forall i \in [h],
\end{equation*}
with an error term fulfilling $|\eps_2(i)| \leq \frac{i}{n_i}$. Since we set $p_0 = 0 = \frac{1}{\lambda+2}\left(1-(\nicefrac{-1}{\lambda+1})^0\right)$, the induction base holds. For the induction hypothesis we use Lemma \ref{lem:match_val_distr_U}.

\begin{gather*}
\begin{aligned}
	p_i &=  \frac{1-p_{i-1}}{\lambda+1} +\eps_1(i)\\
		&= \frac{1}{\lambda+1} - \frac{1}{\lambda+1} \left(\frac{1}{\lambda+2}\left(1 - \left(\frac{-1}{\lambda+1}\right)^{i-1} \right) +\eps_2(i{-}1) \right)+\eps_1(i)\\
		&= \frac{1}{\lambda+1} - \frac{1}{\lambda+2}\left(\frac{1}{\lambda+1} + \left(\frac{-1}{\lambda+1}\right)^i \right) +\frac{\eps_2(i{-}1)}{\lambda+1}+\eps_1(i)\\
		&= \frac{1}{\lambda+1} - \frac{1}{(\lambda+2)(\lambda+1)} - \frac{1}{\lambda+2}\left(\frac{-1}{\lambda+1}\right)^i +\underbrace{\frac{\eps_2(i{-}1)}{\lambda+1}+\eps_1(i)}_{\eqqcolon \eps_2(i)}\\
		&= \frac{1}{\lambda+2} \left(1 - \left(\frac{-1}{\lambda+1}\right)^i\right) +\eps_2(i),
\end{aligned}
\end{gather*}
with
\begin{equation*}
	|\eps_2(i)|\leq \frac{|\eps_2(i{-}1)|}{\lambda+1}+|\eps_1(i)|\leq \frac{i{-}1}{\lambda n_{i-1}}+\frac{1}{n_i} =\frac{i}{n_i},
\end{equation*}
where we used $n_i=\lambda n_{i-1}$. Analogously, applying Lemma \ref{lem:match_val_distr_U} to the closed-form of $p_h$, we obtain 
\begin{align}
p_{h+1} &= \frac{1}{\gamma_1+1} \left(\frac{\lambda+1}{\lambda+2} \left(1 - \left(\frac{-1}{\lambda+1}\right)^{h+1} \right) \right)\nonumber \\ &+\eps_2(h{+}1), \label{eq:pbar1}\\
p_{h+2} &=\frac{1}{(\gamma_2+1)(\gamma_1+1)(\lambda+2)}\left(\gamma_1(\lambda+2) + 1 - \left(\frac{-1}{\lambda+1}\right)^h \right) \nonumber \\ 
&+\eps_2(h{+}2), \label{eq:pbar2}\\
p_{h+3} &=\frac{(\gamma_2\gamma_1+\gamma_2+1)(\lambda+2)-1 + \left(\frac{-1}{\lambda+1}\right)^h}{(\gamma_3+1)(\gamma_2+1)(\gamma_1+1)(\lambda+2)}\nonumber \\ &+\eps_2(h{+}3). \label{eq:pbar3}
\end{align}

\section{Detailed algebraic computations}
\label{sec:app_alg_comp}

\noindent In this section, we present the detailed computations to receive our upper bound formula which is optimized at the end of Section~\ref{sec:proof}. We recall the value of the optimal solution of our instance:
\begin{align}
	2 \OPT	&=|U|+|V|=2|U| =2\sum_{i=0}^{h+3} |U_i| \nonumber\\
&= 2\, \Bigg( k\left(\frac{\lambda^{h+1}-1}{\lambda-1}\right)+k\lambda^h\underbrace{(\gamma_1+\gamma_1\gamma_2+\gamma_1\gamma_2\gamma_3)}_{\eqqcolon \bar{\gamma}}\Bigg)\nonumber\\
&=\frac{2k\lambda^h}{\lambda-1}\left(\lambda-\nicefrac{1}{\lambda^h}+\bar{\gamma}(\lambda-1)\right).
\end{align}
The algorithmic solution value is given by
\begin{align*}
	2\ALG 	&= \sum_{i=0}^{h+2} q_i|V_i| + |U\setminus A| + p_A|A| + 2 \rho |A| \\
			&= 2\rho|A|+\sum_{i=0}^{h+2} |U_i|+\sum_{i=0}^{h-1}q_{i}|V_i|\\
			&+q_h|V_h|+q_{h+1}|V_{h+1}|+q_{h+2}|V_{h+2}|+p_A|A|.
\end{align*}
We insert the values of $p_i$, $q_i$ and the cardinalities $n_i$:
\begin{align}
	2 \ALG =&\ 2\rho|A|\label{eq:o2t}\\
			&+\frac{k\lambda^h}{\lambda-1}\left(\lambda-\nicefrac{1}{\lambda^h}+\gamma_1(\gamma_2+1)(\lambda-1)\right)\label{eq:o1}\\
			&+\sum_{i=0}^{h-1} n_i(p_{i+1}-\eps_2(i+1)) +\sum_{i=0}^{h-1}n_i\eps_3(i)\label{eq:o4}\\
			&+n_h (p_{h+1} - \eps_2(h{+}1)) +n_h\eps_3(h)\label{eq:o5}\\
			&+n_{h+1} (p_{h+2} - \eps_2(h{+}2)) + n_{h+1}\eps_3(h{+}1)\label{eq:o3}\\
			&+n_{h+2} (p_A - \eps_2(h{+}3)) + n_{h+2}\eps_3(h{+}2) \label{eq:o2andahalf}\\
			&+n_A (p_A - \eps_2(h{+}3)) + n_A \eps_2(h{+}3). \label{eq:o2}
\end{align}
 \noindent By equations \eqref{eq:closed_pi}, \eqref{eq:pbar1}-\eqref{eq:pbar3}, the terms $p_i - \eps_2(i)$ do not contain any error terms (they are directly subtracted). Hence all remaining error terms in the above formula sum up to:
\begin{align*}\eps &\coloneqq \left|\sum_{i=0}^{h+2}n_i\eps_3(i)+n_A\eps_2(h{+}3)\right|\\
&\leq \sum_{i=0}^{h+2} (i{+}3)+(h{+}3) \le \frac{(h{+}5)^2+h{+}5}{2}+h{+} 3 \in \Omega(h^2).
\end{align*}
As $2\OPT$ is dominated by $\lambda^h$ and $\lambda>1$, $\frac{\eps}{2 \OPT}$ vanishes for $h\to \infty$. We insert \eqref{eq:pbar1}-\eqref{eq:pbar3} into lines \eqref{eq:o5}-\eqref{eq:o2}. Then, we divide all lines \eqref{eq:o2t}-\eqref{eq:o2} without error terms by $2 \OPT$ and take limits $h \to \infty$:

\begin{align*}
	\frac{\eqref{eq:o1}}{2\OPT}&=\frac{k\lambda^h}{\lambda-1}\frac{\lambda-1}{k\lambda^h}\frac{\left(\lambda-\nicefrac{1}{\lambda^h}+\gamma_1(\gamma_2+1)(\lambda-1)\right)}{2\left(\lambda-\nicefrac{1}{\lambda^h}+\bar{\gamma}(\lambda-1)\right)}\\
	&\to \frac{\left(\lambda+\gamma_1(\gamma_2+1)(\lambda-1)\right)}{2\left(\lambda+\bar{\gamma}(\lambda-1)\right)}=\eqref{eq:p1},
\end{align*}
\begin{align*}
	\frac{\eqref{eq:o4}}{2\OPT}&=\frac{k\lambda^h}{\lambda+2}\frac{\lambda-1}{k\lambda^h}\frac{\left(\frac{1-\nicefrac{1}{\lambda^h}}{\lambda-1}-\frac{(\nicefrac{-1}{\lambda+1})^h-\nicefrac{1}{\lambda^h}}{2\lambda+1}\right)}{2\left(\lambda-\nicefrac{1}{\lambda^h}+\bar{\gamma}(\lambda-1)\right)}\\
	&\to \frac{\lambda-1}{\lambda+2}\frac{\left(\frac{1}{\lambda-1}\right)}{2\left(\lambda+\bar{\gamma}(\lambda-1)\right)}\\ 
	&=\frac{1}{2(\lambda+2)\left(\lambda+\bar{\gamma}(\lambda-1)\right)}=\text{part of } \eqref{eq:p45},
\end{align*}
\begin{gather*}
	\begin{aligned}
	\frac{\eqref{eq:o5}}{2\OPT}
	&=\frac{(\lambda+1)\left(1-(\nicefrac{-1}{\lambda+1})^{h+1}\right) }{(\lambda+2)(\gamma_1+1)}\cdot \lambda^hk\frac{(\lambda-1)}{2 \lambda^hk \left(\lambda-\nicefrac{1}{\lambda^h}+\bar{\gamma}(\lambda-1)\right)}\\
	&\to\frac{(\lambda+1) (\lambda-1)}{2 (\lambda+2)(\gamma_1+1)\lambda+\bar{\gamma}(\lambda-1)} = \text{part of } \eqref{eq:p45},
	\end{aligned}
\end{gather*}
\begin{gather*}
	\begin{aligned}
	\frac{\eqref{eq:o3}}{2\OPT}&=
	\frac{\gamma_1(\lambda+2)+1-(\nicefrac{-1}{\lambda+1})^h}{(\gamma_2+1)(\gamma_1+1)(\lambda+2)} \cdot \gamma_1\lambda^hk\frac{(\lambda-1)}{2 \lambda^hk\left(\lambda-\nicefrac{1}{\lambda^h}+\bar{\gamma}(\lambda-1)\right)}\\	
	&\to \frac{\gamma_1(\lambda+2)+1}{(\gamma_2+1)(\gamma_1+1)(\lambda+2)} \frac{\gamma_1(\lambda-1)}{2\left(\lambda+\bar{\gamma}(\lambda-1)\right)}=\eqref{eq:p3},
	\end{aligned}
\end{gather*}
\begin{gather*}
	\begin{aligned}
	\frac{\eqref{eq:o2andahalf}+\eqref{eq:o2}}{2\OPT}&=\frac{(\lambda-1)(\lambda^hk\gamma_1\gamma_2+\lambda^hk\gamma_1\gamma_2\gamma_3)}{2 \lambda^hk\left(\lambda-\nicefrac{1}{\lambda^h}+\bar{\gamma}(\lambda-1)\right)} \, \cdot \\
	&\qquad \frac{(\gamma_2\gamma_1+\gamma_2+1)(\lambda+2)-1+(\nicefrac{-1}{\lambda+1})^h}{(\gamma_3+1)(\gamma_2+1)(\gamma_1+1)(\lambda+1)}\\
	&=\frac{(\lambda-1)\gamma_1\gamma_2}{2\left(\lambda-\nicefrac{1}{\lambda^h}+\bar{\gamma}(\lambda-1)\right)}\, \cdot \\
	&\qquad \frac{\gamma_2(\gamma_1+1)(\lambda+2)+(\lambda+1)+(\nicefrac{-1}{\lambda+1})^h}{(\gamma_2+1)(\gamma_1+1)(\lambda+1)}\\
	&\to\frac{\gamma_1\gamma_2(\lambda-1)}{2(\lambda+\gamb(\lambda-1))} \cdot\frac{\gamma_2(\gamma_1+1)(\lambda+2)+(\lambda+1)}{(\gamma_2+1)(\gamma_1+1)(\lambda+2)}\\
	&=\eqref{eq:p2},
	\end{aligned}
\end{gather*}
\begin{gather*}
	\begin{aligned}
	\frac{\eqref{eq:o2t}}{2\OPT}&\leq\frac{|A|(1-\exp(-(1-p_A))+\nicefrac{2}{|A|})}{\OPT}\\
	&=\frac{2}{\OPT}+\frac{k\lambda^h\gamma_1\gamma_2\gamma_3(\lambda-1)}{k\lambda^h\left(\lambda-\nicefrac{1}{\lambda^h}+\bar{\gamma}(\lambda-1)\right)} \cdot \Bigg[1{-}\exp\Bigg(-1 \, + \\
	&\qquad \frac{\gamma_2(\gamma_1+1)(\lambda+2)+(\lambda+1)+(\nicefrac{-1}{\lambda+1})^h}{(\gamma_3+1)(\gamma_2+1)(\gamma_1+1)(\lambda+2)}+\eps_2(h+3)\Bigg)\Bigg]\\
	&\to \frac{\gamma_1\gamma_2\gamma_3(\lambda-1)}{\left(\lambda+\bar{\gamma}(\lambda-1)\right)} \, \cdot \\
	&\qquad \left[1{-}\exp\left(-1+\frac{\gamma_2(\gamma_1+1)(\lambda+2)+(\lambda+1)}{(\gamma_3+1)(\gamma_2+1)(\gamma_1+1)(\lambda+2)}\right)\right]\\
	&=\eqref{eq:p2t}.
	\end{aligned}
\end{gather*}
As $h\to\infty$, our upper bound on the competitive ratio of ALG therefore approaches
\begin{align}
&\frac{\lambda+\gamma_1(\gamma_2+1)(\lambda-1)}{2(\lambda+\gamb(\lambda-1))}\label{eq:p1} \\
+&\frac{\lambda^2+\gamma_1}{2(\lambda+2)(\gamma_1+1)(\lambda+\gamb(\lambda-1))}\label{eq:p45}\\
+&\frac{\gamma_1(\lambda-1)}{2(\lambda+\gamb(\lambda-1))} \cdot\frac{\gamma_1(\lambda+2)+1}{(\gamma_2+1)(\gamma_1+1)(\lambda+2)}\label{eq:p3}\\
+&\frac{\gamma_1\gamma_2(\lambda-1)}{2(\lambda+\gamb(\lambda-1))} \cdot\frac{\gamma_2(\gamma_1+1)(\lambda+2)+(\lambda+1)}{(\gamma_2+1)(\gamma_1+1)(\lambda+2)}\label{eq:p2}\\
+&\frac{\gamma_1\gamma_2\gamma_3(\lambda-1)}{(\lambda+\gamb(\lambda-1))} \cdot \Bigg[1-\exp\Bigg(-\frac{1}{\gamma_3+1} \, \cdot \nonumber\\
&\qquad \left(\gamma_3+\frac{\gamma_1(\lambda+2)+1}{(\gamma_2+1)(\gamma_1+1)(\lambda+2)}\right)\Bigg)\Bigg]. \label{eq:p2t}
\end{align}

\end{document}